\pdfoutput=1
\documentclass[12pt,reqno]{amsart}

\numberwithin{equation}{section}
\usepackage[tt=false]{libertine}
\usepackage{mathtools}

\usepackage{amssymb,mathrsfs}
\usepackage[varbb]{newpxmath}

\let\savedbigtimes\bigtimes
\let\bigtimes\relax
\usepackage{mathabx} 
\let\bigtimes\savedbigtimes

\usepackage[margin=1in]{geometry}
\usepackage{graphicx}
\usepackage{enumerate}
\usepackage{bbm}
\usepackage{hyperref,color}
\usepackage[capitalize,nameinlink]{cleveref}
\usepackage[dvipsnames]{xcolor}
\hypersetup{
	colorlinks=true,
	pdfpagemode=UseNone,
    citecolor=OliveGreen,
    linkcolor=NavyBlue,
    urlcolor=black,
	pdfstartview=FitW
}
\usepackage{appendix}
\crefname{appsec}{Appendix}{Appendices}
\usepackage{tikz}
\usepackage{bm}
\usepackage{mathtools}

\newtheorem{theorem}{Theorem}[section]
\newtheorem{proposition}[theorem]{Proposition}

\newtheorem{corollary}[theorem]{Corollary}

\newtheorem{fact}[theorem]{Fact}

\theoremstyle{definition}
\newtheorem{definition}[theorem]{Definition}

\newtheorem*{assumption*}{Assumption}
\newtheorem{remark}[theorem]{Remark}

\crefname{lemma}{Lemma}{Lemmas}
\crefname{theorem}{Theorem}{Theorems}
\crefname{definition}{Definition}{Definitions}
\crefname{fact}{Fact}{Facts}
\crefname{claim}{Claim}{Claims}
\crefname{proposition}{Proposition}{Propositions}

\newcommand{\E}{\mathbb{E}}
\newcommand{\Var}{\mathrm{Var}}

\renewcommand{\epsilon}{\varepsilon}

\newcommand{\N}{\mathbb{N}}
\newcommand{\Q}{\mathbb{Q}}
\newcommand{\R}{\mathbb{R}}

\renewcommand{\P}{\mathbb{P}}

\newcommand{\beq}{\begin{equation}}
\newcommand{\eeq}{\end{equation}}

\renewcommand\Pr{\mathbf{P}}
\newcommand{\abs}[1]{\left| #1 \right|}

\usepackage{appendix}
\crefname{appsec}{Appendix}{Appendices}
\newcommand{\LDLR}{\textnormal{\textsf{LDLR}}}
\usepackage{xspace}

\let\Erdos\erdos

\let\Renyi\renyi

\definecolor{blueish}{HTML}{1cb5bd} 
\definecolor{greenish}{HTML}{b5b51c} 
\definecolor{pinkish}{HTML}{EE51DE}

\newcommand\cnote[1]{\textsf{\small\color{blueish}{[Chris: #1]}}}
\newcommand\abnote[1]{\textsf{\small\color{greenish}{[Andrej: #1]}}}

\newtheorem{construction}{Construction}{\bfseries}{\normalfont}

\newtheorem{bonjecture}[theorem]{Conjecture}{\bfseries}{\itshape}

\begin{document}

\title[Low-degree Security of the Planted Random Subgraph Problem]{Low-degree Security of the \\ Planted Random Subgraph Problem}

\author[A. Bogdanov, C. Jones, A. Rosen, I. Zadik]{Andrej Bogdanov$^{\star}$, Chris Jones$^\circ$, Alon Rosen$^{\circ \dagger}$ \\ and Ilias Zadik$^\mathsection$}
\thanks{\raggedright$^\star$University of Ottawa.;
$^\circ$Bocconi University;
$^\dagger$Reichman University;
$^\mathsection$ Yale University. \\ Emails: \texttt{abogdano@uottawa.ca}; 
 \texttt{chris.jones@unibocconi.it}; \texttt{alon.rosen@unibocconi.it}; \texttt{ilias.zadik@yale.edu}}

\maketitle

\begin{abstract}
The planted random subgraph detection conjecture of Abram et al. (TCC 2023) asserts the pseudorandomness of a pair of graphs $(H, G)$, where $G$ is an \Erdos-\Renyi random graph on $n$ vertices, and $H$ is 
a random induced subgraph of $G$ on $k$ vertices.
Assuming the hardness of distinguishing these two distributions (with two leaked vertices), Abram et al. construct communication-efficient, computationally secure (1) 2-party private simultaneous messages (PSM)  and (2) secret sharing for forbidden graph structures. 

\smallskip

We prove the low-degree hardness of detecting planted random subgraphs all the way up to $k\leq n^{1 - \Omega(1)}$.  This improves over Abram et al.'s analysis for $k \leq n^{1/2 - \Omega(1)}$.  The hardness extends to $r$-uniform hypergraphs for constant $r$.


\smallskip

Our analysis is tight in the distinguisher's degree, its advantage, and in the number of leaked vertices.  Extending the constructions of Abram et al, we apply the conjecture towards (1) communication-optimal multiparty PSM protocols for random functions and (2) bit secret sharing with share size $(1 + \epsilon)\log n$ for any $\epsilon > 0$ in which arbitrary minimal coalitions of up to $r$ parties can reconstruct and secrecy holds against all unqualified subsets of up to $\ell = o(\epsilon \log n)^{1/(r-1)}$ parties.






\end{abstract}

\section{Introduction}
In the planted clique model~\cite{jerrum, kucera} one observes the union of an \Erdos-\Renyi random graph $G_0 \sim G(n,1/2)$ and a randomly placed $k=k_n$-clique $H$, i.e., the graph $G=G_0 \cup H$. The goal of the planted clique detection task is to distinguish between observing $G$ from the planted clique model and $G$ which is simply an instance of $G(n,1/2)$. The planted clique conjecture states that the planted clique instance remains pseudorandom whenever $k \leq n^{1/2 - \Omega(1)}$ up to $n^{-\Omega(1)}$ distinguishing advantage.
Conversely, multiple polynomial-time algorithms can distinguish with high probability whenever $k=\Omega(\sqrt{n})$. Research on the planted clique conjecture has gone hand-in-hand with key developments in average-case complexity theory over the last decades, including spectral and tensor algorithms \cite{alon1998finding, frieze2008new}, lower bound techniques for restricted classes including the sum-of-squares hierarchy \cite{barak-hopkins-etal}, low-degree polynomial methods \cite{hopkins-thesis}, statistical query methods \cite{feldman-grigorescu-reyzin-vempala-xiao} and MCMC methods \cite{jerrum,gamarnik2019landscape,chen2023almost}, and the development of new average-case reductions \cite{brennan-bresler, hirahara-shimizu}. 

At this point, the conjectured hardness of the planted clique problem around $k \approx \sqrt{n}$ stands as a central conjecture in average-case complexity. But despite its popularity, the cryptographic applications have been quite limited, with one exception in the symmetric-key setting proposed by Juels and Peinado~\cite{juels-peinado}.  Recently Abram et al.~\cite{abram2023cryptography} revisited the planted clique problem and showed how it can be useful in the context of secret sharing and secure computation. The authors specifically show that (slight variants of) the planted clique conjecture can be used to construct a computationally secure scheme whose share size is much smaller than the best existing information-theoretically secure scheme.

In order to obtain further improvements to the share size, Abram et al. proposed a new intriguing conjecture similar to planted clique. They start by defining the following general model (also introduced in \cite{huleihel2022inferring}). 

\begin{definition}\emph{(Planted (induced) subgraph model\footnote{A similar yet different model where one observes the \emph{union} of a copy of $H$ with an instance of $G(n,1/2)$ has also been recently analyzed in the statistical inference literature \cite{huleihel2022inferring, mossel2023sharp, yu2024counting}. For this work, we solely focus on the ``induced'' variant, where $H$ appears as an induced subgraph of $G.$})} \label{dfn:sub_model} Fix $H$ to be an arbitrary unlabeled subgraph on $k$ vertices. Then $G$ is chosen to be a random $n$-vertex graph where a copy of $H$ is placed on $k$ vertices chosen uniformly at random (as an induced subgraph on the $k$ vertices), and all edges without both endpoints on the $k$ vertices appear with probability 1/2.
\end{definition}

When $H$ is the $k$-clique, the planted subgraph model becomes exactly the planted clique model. The clique is the most structured graph possible and it is natural to wonder:\begin{center} 
   \emph{could the problem be significantly harder if a different graph $H$ is planted?}
\end{center} Abram et al. suggest studying the planted random subgraph model in which $H$ is an instance of $G(k,1/2).$
An equivalent definition is the following.

\begin{definition}\emph{(Planted random subgraph model)} One observes a pair $(H, G)$, where $G$ is a random $n$-vertex graph and $H$ is a random $k$-subgraph of $G$ with the vertex labels removed.
\end{definition}

Abram et al. make the following interesting conjecture.\begin{bonjecture}\label{conj_main} \emph{(Planted Random Subgraph conjecture \cite{abram2023cryptography})}
The planted random subgraph problem is hard up to advantage  $n^{-\Omega(1)}$ provided $k \leq n^{1 - \Omega(1)}$, with high probability over $H \sim G(k,1/2)$ as $n$ grows to infinity. 
\end{bonjecture}

This stands in contrast to the case that $H$ is a $k$-clique where a computational phase transition is expected to take place at the smaller value $k \approx n^{1/2}$.

Abram et al. confirm the planted random subgraph conjecture in the low-degree analysis framework (to be described below)
but only up to the ``planted clique threshold'' $k \leq n^{1/2 - \Omega(1)}$ (a result also independently proven by Huleihel \cite{huleihel2022inferring}). Their work leaves open the regime $ n^{1/2 - \Omega(1)} \leq k \leq n^{1 - \Omega(1)},$ and in particular the question of whether there is a larger window of hardness for planted random subgraph than for planted clique.


Our main contribution is the confirmation of Conjecture~\ref{conj_main} in the low-degree framework. We prove that the planted random subgraph problem remains hard for low-degree distinguishers of degree at most $o((\log n / \log \log n)^2)$ in the \emph{full range} $k \leq n^{1 - \Omega(1)}$. The degree is best possible up to $\log \log n$ factors, and the analysis extends also to the case of hypergraphs. See Section~\ref{sec:main} for the precise theorem statement.

\subsection{Secret sharing and leakage}
\label{sec:abrametal}

For their intended cryptographic applications Abram et al. rely on a strengthening of the planted random subgraph conjecture which also allows for leaked additional information about the embedding of $H$ in $G$.  It is easiest to motivate these stronger conjectures through their intended application.

A \emph{(partial) access structure} for $k$ parties is a pair of set systems $R, S$ over $\{1, \dots, k\}$, where $R$ is upward-closed, $S$ is downward-closed, and $R, S$ are disjoint.  A bit secret sharing scheme consists of a randomized sharing algorithm that maps the secret bit $s \in \{0,1\}$ into $k$ shares so that sets in $R$ can reconstruct $s$ from their shares with probability one, while sets in $S$ cannot distinguish $s=0$ or $s=1$.

In a forbidden graph access structure, $R$ is the edge-set of a graph and $S$ is the union of its complement $\{\{u, v\} \not\in R\colon u \neq v \in [k]\}$ and the set $[k]$ of vertices.  Abram et al. propose the following secret sharing scheme for any such structure: 

\begin{construction}
\label{cons:abrametal}
Forbidden graph secret sharing:
\begin{enumerate}
\item The dealer samples a random $n$-vertex graph $G$ and remembers a secret $k$-vertex subgraph $H$ of it randomly embedded via $\phi\colon V(H) \to V(G)$.  

\item The dealer publishes the pair $(H_s, G)$, where $H_s$ is a $k$-vertex graph with adjacency matrix
\begin{equation}
\label{eq:defhs}
H_s(u, v) = \begin{cases}
H(u, v) \oplus s, &\text{if $\{u, v\} \in R$} \\
\text{a random bit}, &\text{otherwise}.
\end{cases} 
\end{equation}

\item  The share of party $v$ is the value $\phi(v) \in [n]$.  
\end{enumerate}

\noindent If $\{u, v\} \in R$, the parties reconstruct by calculating 
\begin{equation}
\label{eq:reconstruct}
H_s(u, v) \oplus G(\phi(u), \phi(v)) = H(u, v) \oplus G(\phi(u), \phi(v)) \oplus s = s. 
\end{equation}
\end{construction}

Secrecy requires that the joint distribution $(H_s, G, \phi(u), \phi(v))$ of the public information and the shares is indistinguishable between $s = 0$ and $s = 1$ provided $\{u, v\} \in S$. In the absence of the ``leakage'' $(\phi(u), \phi(v))$ this is a consequence of the planted random subgraph conjecture (Conjecture \ref{conj_main}).

To handle the leakage, we consider the following generalization. Two parties $\{u,v\} \in S$ know the location of their edge $H(u,v) = G(\phi(u), \phi(v))$ in $G$, which could potentially be useful to search for the ``local structure of $H$'' around their edge.
The new conjecture posits that if $u$ and $v$ have this additional information, they still cannot distinguish whether $H$ is planted. 
With an eye towards stronger security we state it below for a general $\ell$.

\begin{bonjecture}\label{conj:leakage}
\emph{(Planted random subgraph conjecture with $\ell$-vertex leakage)} 
With high probability over $H \sim G(k,1/2)$, 
the following two distributions are 
$n^{-\Omega(1)}$-indistinguishable in polynomial time for all subsets $L = \{u_1, \dots, u_\ell\} \subseteq V(H)$ of size $\ell$:
\begin{enumerate}
    \item (planted) $(H, G, \phi(u_1), \dots, \phi(u_\ell))$ where we choose uniformly at random an injective function $\phi :[k] \to [n]$ and embed $H$ into $G$ on the image of $\phi$. The remaining edges of $G$ are sampled randomly.
    \item (model) $(H, G, \phi(u_1), \dots, \phi(u_\ell))$ where we choose uniformly at random an injective function $\phi : L \to [n]$ and embed the subgraph of $H$ on $L$ into $G$ on the image of $\phi$. The remaining edges of $G$ are sampled randomly.
\end{enumerate}
\end{bonjecture}

Assuming this conjecture with $\ell = 2$, given $(\phi(u), \phi(v))$ for $\{u,v\} \in S$, we claim that both $(H_0, G)$ and $(H_1, G)$ are pseudorandom and hence indistinguishable:  As $\{u, v\} \in S$, the $(u, v)$-th bits of $H_0$ and $H_1$ in $(H_s, G)$ are independent of all the others and cannot be used to distinguish.  Once the $(u, v)$-th bits of $H_0$ and $H_1$ are removed, both $(H_0, G)$ and $(H_1, G)$ become identically distributed to the planted $(H, G)$ with its $(u, v)$-th bit removed. By the conjecture, this model is indistinguishable from a uniformly random string.

The share size in this scheme is $(1 + o(1))\log k$.  In contrast, the most compact known forbidden graph scheme with perfect security has shares of size $\exp \tilde{\Theta}(\sqrt{\log k})$~\cite{liu-vaikutanathan-wee,applebaumbeimeletal2019}.
Statistical security requires shares of size $\log k - O(1)$ when $R$ is the complete graph~\cite{abram2023cryptography}.  It is not known if computational security is subject to the same limitation.

Under the $\ell$-vertex leakage assumption the secrecy holds not only against pairs of parties that are not an edge in $R$, but also against all independent sets up to size $\ell$, i.e., 
\[ S = \{I\colon \text{$I$ is an independent set of $R$ and $\abs{I} \leq \ell$}\}. \]

By passing to $r$-hypergraphs instead of graphs, we naturally extend the construction to $R$ which is an arbitrary subset of at most $r$ parties, with security against all size-$\ell$ independent sets of $R$ (see Construction~\ref{cons:atmostr} below).  The most compact known perfectly secure forbidden $r$-hypergraph scheme has share size $\exp \tilde{\Theta}(\sqrt{r\log k})$~\cite{liu-vaikutanathan-wee} whereas our share size is still $(1+o(1))\log k$.

It would be interesting to obtain a provable separation in share size between the computationally secure Construction~\ref{cons:atmostr} and the best possible perfectly secure construction for some access structure.  In Section~\ref{sec:secretsharing} we explain why this is challenging using available methods.

\subsection{Private simultaneous messages (PSM)}  In a PSM, Alice and Bob are given inputs $x, y$ to a public function $F\colon [k]^2 \to \{0, 1\}$.  They calculate messages $\phi(x), \phi(y)$ which are securely forwarded to Carol.  Carol needs to output the value $F(x, y)$ without learning any information about $x$ and $y$ beyond this value.

Abram et al. propose the following PSM protocol.  In a setup phase, $F$ viewed as a bipartite graph is randomly embedded into an otherwise random host graph $G$ via $\phi$.  The graph $G$ is given to Carol and the embedding $\phi$ is given to Alice and Bob.  Carol outputs $G(\phi(x), \phi(y))$ which must equal $F(x, y)$.  


Abram et al. argue that this protocol is ``secure'' for a $(1 - o(1))$-fraction of functions $F$ under Conjecture~\ref{conj:leakage} with leakage $\ell = 2$.  Their security definition appears to additionally assume that the choice of inputs $(x, y)$ is independent of the function $F$.  In contrast, our security definition in Section~\ref{sec:psm} allows for Alice and Bob to choose their inputs jointly from some distribution that depends on the description of $F$.  This is more natural for potential  cryptographic applications; Alice and Bob should not be expected to commit to their input before they know which function they are computing.
We extend our low-degree analysis to support this stronger notion of security.

Messages in this protocol are of length $\log n = (1 + \epsilon)\log k$.
In contrast, perfect security is known to require combined message length $\abs{\phi(x)} + \abs{\phi(y)} \geq (3 - o(1))\log k$~\cite{feige-kilian-naor, applebaum-holenstein-mishra-shayevitz} (but it is not known if statistical security is subject to the same bound).


The $r$-hypergraph variant of the conjecture with leakage $\ell = r$ gives PSM security for $r$-party protocols also with message size $\log n = (1 + \epsilon) \log k$ (Section~\ref{sec:psm}).  Even without a security requirement the message size must be at least $(1 - o(1))\log k$ for the protocol to be correct on most inputs.


\subsection{Low-degree lower bounds}
We provide evidence for these conjectures in the form of lower bounds against the low-degree polynomial computational model (see e.g., \cite{kunisky2019notes} and references therein).
In this model, fixing a parameter $D=D_n$, the distinguishing algorithm is allowed to compute an arbitrary degree-$D$ polynomial function of the bits of the input over the field $\R$.
The algorithm succeeds if the value of the polynomial is noticeably different between the random and planted models.
Degree-$D$ polynomials serve as a proxy for $n^{O(D)}$ time computation since a degree-$D$ polynomial in $\text{poly}(n)$ input bits can be evaluated by brute force in time $n^{O(D)}$ (ignoring numerical issues).

Surprisingly, for noise-robust\footnote{Noise-robustness means that the planted structure is resilient to small random perturbations \cite{hopkins-thesis,holmgren2020counterexamples}.} hypothesis testing problems it has been conjectured that whenever all degree-$D$ polynomials with $D=O(\log n)$ fail (formally, no polynomial strongly separates the two distributions \cite[Section 7]{cojaoghlan2022statistical}), then no polynomial-time distinguisher succeeds. This is now known as the ``low-degree conjecture'' of Hopkins \cite{hopkins-thesis}.
Based on this heuristic, a provable failure of $O(\log n)$-degree polynomials to strongly separate the two distributions provides a state-of-the-art prediction of the hard and easy regimes for the problem of interest.

It should be noted that there exists a certain weakness in existing low-degree hardness evidence for the planted clique problem, which also applies to our lower bound for the planted random subgraph problem (and that of \cite{abram2023cryptography}). Both planted clique and planted random subgraph technically do not satisfy the noise-robust assumption of the low-degree conjecture because the planted isomorphic copy of $H$ in the graph $G$ is not robust to small perturbations of $G$ (if $0.01$ fraction of the edges of $G$ are randomly flipped then the copy of $H$ will be destroyed).
Noise-robustness is an important assumption; in fact, in a handful of carefully chosen noise-free problems, low-degree methods are provably weaker than other brittle polynomial-time methods such as Gaussian elimination or lattice-basis reduction techniques \cite{zadik2022lattice}. That being said, the existing techniques do not appear applicable to graph settings such as planted clique or the planted random subgraph model.

\section{Our result} 
\label{sec:main}
Let $H$ be an $r$-uniform hypergraph over vertex set $[k]$ chosen uniformly at random
(i.e., each $r$-hyperedge between the vertices of $[k]$ is included independently with probability half). Let $L \subseteq V(H)$ of size $\ell$.
Let $\P_{H, L}$ and $\Q_{H, L}$ be the following distributions over $r$-uniform hypergraphs $G$ with vertex set $[n]$, where $n \geq k \geq \ell$:
\begin{enumerate}
\item In the \emph{planted distribution} $\P_{H, L}$, an injective map $\phi\colon [k] \to [n]$ is chosen uniformly at random among all injective maps conditioned on $\phi(u) = u$ for $u \in L$.  The hyperedges of $G$ are
\[ G(u_1, \dots, u_r) = \begin{cases}
H(\phi^{-1}(u_1), \dots, \phi^{-1}(u_r)), &\text{if $\phi^{-1}(u_1), \dots, \phi^{-1}(u_r)$ exist} \\
\text{a random bit},  &\text{otherwise}.
\end{cases}
\]
\item In the \emph{null distribution} $\Q_{H, L}$, 
the hyperedges of $G$ are
\[ G(u_1, \dots, u_r) = \begin{cases}
H(u_1, \dots, u_r), &\text{if $u_1, \dots, u_r \in L$} \\
\text{a random bit}, &\text{otherwise}.
\end{cases}
\]
\end{enumerate}
Uniform $r$-hypergraphs on $n$ vertices are represented by their adjacency maps $\binom{[n]}{r} \to \{\pm 1\}$, with $-1$ and $1$ representing the presence and absence of a hyperedge, respectively.

In words, the hypergraph $G\sim \P_{H, L}$ drawn from the planted model has the public hypergraph $H$ embedded into a uniform choice of $k$ vertices, and is otherwise purely random. However, the location of $L \subseteq V(H)$ is fixed and public information.
The hypergraph $G \sim \Q_{H, L}$ drawn from the random model copies the subgraph of $H$ on $L$, but it does not use the part of $H$ outside of $L$; all remaining edges of the graph are chosen purely at random. 
Note that in both models, the marginal distribution of $G$ is a uniformly random hypergraph,
but distinguishers know $H$ and $L$.

In the case $r = 2$ of graphs, there is a slight difference between the distributions $\P_{H, L}, \Q_{H, L}$ and those described in the Introduction, namely that we have imposed the condition $\phi(u) = u$ on the leaked vertices in $L$.
This condition is without loss of generality, and in particular, it does not affect the complexity of distinguishing $\P_{H, L}$ from $\Q_{H, L}$.

Following the low-degree framework \cite{kunisky2019notes}, we consider the degree-$D$-likelihood ratio $\mathcal{LR}_D(H, L)$,
\[ \mathcal{LR}_D(H, L) = \sup\limits_{\substack{p \in \R[G(\vec{u}) : \vec{u} \in \binom{[n]}{r}]\\ \deg p \le D}} \mathrm{Adv}_p(H, L) \] 
where
\[ \mathrm{Adv}_p(H, L) = \frac{\E_{\P_{H, L}}[p(G)] - \E_{\Q_{H, L}}[p(G)]}{\sqrt{\Var_{\Q_{H, L}}[p(G)]}}. \]
Here $p \in \R[G(\vec{u}) : \vec{u} \in \binom{[n]}{r}]$ denotes a multivariate polynomial in the quantities $G(u_1, \dots, u_r)$ for $(u_1,\dots, u_r) \in \binom{[n]}{r}$ with degree at most $D$. $\mathcal{LR}_D(H, L)$ measures the best advantage of a degree-$D$ polynomial distinguisher that can arbitrarily preprocess $H$ and knows $L$.
Whenever $\mathcal{LR}_D(H, L)=o(1)$ then no $D$-degree polynomial can achieve strong separation between $\P_{H, L}$ and $\Q_{H, L}$ \cite[Section 7]{cojaoghlan2022statistical}.

To gain intuition on the performance of low-degree polynomials, let us start with the simplest one, which is the bias of the edges of the hypergraph $G$:
\[ p(G) = \sum_{1 \leq u_1 < \cdots < u_r \leq n} G(u_1, \dots, u_r). \]
Assume for simplicity that $L = \emptyset$.
It holds by direct expansion,
\begin{align*}
\E_{\P_{H}}[p(G)] &= \sum_{1 \leq u_1 < \dots < u_r \leq k} H(u_1, \dots, u_r) \\
\E_{\Q_{H}}[p(G)] &= 0 \\
\Var_{\Q_{H}}[p(G)] &= \tbinom{n}{r}.
\end{align*}
The likelihood ratio is
\[ \mathrm{Adv}_p(H) = \Theta\biggl(\frac{\E_{\P_{H}}[p(G)]}{n^{r/2}}\biggr). \]
As $\E_{\P_{H}}[p(G)]$ is a sum of the 
$\binom{k}{r}$ hyperedge indicators for $H$, $\E_{\P_{H}}[p(G)]$ would have value $\pm \Theta(k^{r/2})$ for a typical choice of $H$, resulting in an advantage of $\Theta((k/n)^{r/2})$ (after optimizing between $p(G)$ or $-p(G)$). The advantage is $o(1)$ when $k \leq n^{1-\Omega(1)}$ and therefore the distinguisher fails in this regime. Yet, when $k=\Theta(n)$ the calculation suggests the count distinguisher succeeds with $\Omega_r(1)$ probability which indeed can be confirmed by being a bit more careful in the above analysis. 
Our main theorem shows that other low-degree polynomials cannot substantially improve upon the edge-counting distinguisher.

\begin{theorem}\label{thm:main}
Assume for some $p \in \mathbb{N}$ and constant $\epsilon>0$,
the following bounds hold on the size of $H$, $k$, the leakage number $\ell$ and the degree $D$:

\begin{enumerate}
    \item $k \leq (n - \ell) n^{-\epsilon} / 24 p^2 D^2 + \ell$
    \item $\ell \leq \min\{k, \;\epsilon^{1/(r-1)} r(\log n)^{1/(r-1)}/40\}$ and,
    \item $D \leq \epsilon^3 \left(\log n\right)^{r/(r-1)}/\bigl(\frac{r}{r-1}\log \log n\bigr).$ 
\end{enumerate}
Then for any $L \subseteq [k]$ with $|L| = \ell$,
\[ \left(\E_H \mathcal{LR}_D(H, L)^{2p}\right)^{1/p} \leq  
\frac{2^{\binom{\ell}{r - 1}} n^{-\epsilon}}{1 - n^{-\epsilon/2}} + \exp\left(-\Omega\left( r(\epsilon \log n)^{1+1/(r-1)}\right)\right).
\]
\end{theorem}
In particular, for $p=1$, $\ell=o( (\log n)^{1/(r-1)}), $ and $\epsilon=\Omega(1)$ \[\E_H \mathcal{LR}_D(H, L)^2=n^{-\epsilon + o(1)}.\]

The bound is tight in the following ways:
\begin{enumerate}
\item \textbf{Degree:}  The bound on $D$ is optimal (for constant $\epsilon$) up to a factor of $O(\log \log n)$.  A degree-$O((r\log n)^{r/(r-1)})$ distinguisher with high advantage and time complexity $2^{O((r\log n)^{1/(r-1)})}$ exists.  This is the algorithm that looks for the presence of a subgraph in $G$ that is identical to the one induced by the first $O(r^{r/(r-1)} (\log n)^{1/(r-1)})$ vertices in $H$. 
\item \textbf{Leakage:}  When $\binom{\ell}{r-1} \geq \log(2n)$ the distinguishing advantage is constant (for any $k > \ell$).  The distinguisher that looks for the existence of a vertex in $G$ whose adjacencies in $L$ match those of an arbitrary vertex in $H$ outside $L$ has constant advantage, degree $\binom{\ell}{r-1}$, and time complexity $O(n \binom{\ell}{r-1})$.

\item \textbf{Advantage:}  The edge-counting distinguisher described above has advantage $(k/n)^{r/2} = n^{-\epsilon r/2}$.  Our proof can show a matching lower bound in the absence of leakage. When leakage is present, assuming $\ell > r - 1$, the linear distinguisher
\[ \mathrm{sign} \sum_{v \not\in L} G(1, \dots, r - 1, v) = \mathrm{sign} \sum_{v \not\in L} H(1, \dots, r - 1, v) \]
has squared advantage $\Omega((k - \ell)/(n - \ell)) = \Omega(n^{-\epsilon})$ which matches the theorem statement.
\end{enumerate}

\subsection{Our proof} 

Abram et al. obtain their result as a consequence of a worst-case bound for arbitrary planted $H$:  They prove that for all graphs $H$ with $k \leq n^{1/2-\epsilon}$ vertices,
\[ \mathcal{LR}_D(H, L) \leq o(1) \,.\]
As $k = n^{1/2}$ is tight for clique their method cannot prove a better bound.  In contrast, we average the likelihood ratio over the choice of $H$, showing that $\E_H[\mathcal{LR}_D(H, L)^2]$ is small all the way up to $k \leq n^{1-\epsilon}$. By taking the expectation over $H$, we introduce extra cancellations that are necessary to obtain the stronger bound.

By Markov's inequality
\[\Pr_H[\mathcal{LR}_D(H, L)^2 \geq \eta] \leq \frac{\E_H[\mathcal{LR}_D(H, L)^2]}{\eta}\,.\]
A vanishing expectation implies concentration, namely $\mathcal{LR}_D(H, L) = o(1)$ for a $1 - o(1)$ fraction of $H$.

The above calculation bounds the advantage for a fixed leakage set $L$.
In order to bound the advantage of an arbitrary set $L$ for the cryptographic applications, we also bound the higher moments of $\mathcal{LR}_D(H, L)$.
%
%
%
Using $p = \ell \log n$ and applying Markov's inequality with $\eta = 4n^{-\epsilon+o(1)}$
\begin{align*}
    \Pr_H[\mathcal{LR}_D(H, L)^2 &\geq \eta] \leq \frac{\E_H[\mathcal{LR}_D(H, L)^{2p}]}{\eta^{p}}\\
    &\leq \left(\frac{n^{-\epsilon + o(1)}}{\eta}\right)^p\\
    &= 4^{-\ell \log n} \leq \frac{1}{n\binom{n}{\ell}}\,.
\end{align*}
Taking a union bound over the $\binom{k}{\ell}$ choices for $L$, we can deduce the stronger result that no leakage set $L$ can attain advantage $\eta$:
\[
    \Pr_H\biggl[\max\nolimits_{\substack{L \subseteq V(H)\\|L|=\ell}} \mathcal{LR}_D(H,L)^2 \geq 4n^{-\epsilon + o(1)} \biggr] \leq o(1)\,.
\]
We summarize the final bound on the low-degree advantage for Conjecture \ref{conj:leakage} as the following corollary, which includes the parameters.
\begin{corollary}\label{cor:max-bd}
    For all $p \in \mathbb{N}$ and $\eta > 0$,
    \begin{multline*}
    \Pr_H\biggl[\max\nolimits_{\substack{L \subseteq V(H)\\|L|=\ell}} \mathcal{LR}_D(H,L)^2 \geq \eta \biggr] \\
    \leq \binom n \ell \eta^{-p}\left( \frac{2^{\binom{\ell}{r - 1}} n^{-\epsilon}}{1 - n^{-\epsilon/2}} + \exp\left(-\Omega\left( r(\epsilon \log n)^{1+1/(r-1)}\right)\right) \right)^p \,.
    \end{multline*}
\end{corollary}

\section{Proof of Theorem~\ref{thm:main}}

Viewed as an $\binom{n}{r}$-dimensional vector, every $G$ in the support of $\Q_{H, L}$ decomposes as $(G', G_L)$, where $G_L$ is the subgraph of $G$ on $L$ and $G'$ is the remaining part (indexed by $r$-subsets that have at least one vertex in $[n] \setminus L$).

We start by claiming that without loss of generality, all polynomial distinguishers of interest are constant in the coordinates of $G_L$. Indeed, in both the planted $\P_{H, L}$ and null distributions $\Q_{H, L}$, the status of the hyperedges in $L$ is always fixed. As fixing the $L$-indexed inputs can only lower the degree of the distinguishing polynomial $p$, this assumption holds without loss of generality.

In the null $G'$ is simply uniformly random in $\{\pm 1\}^{\binom{[n]}{r} \setminus \binom{L}{r}}$, i.e., $\Q_{H, L}(G', G_L) = \Q(G')$, where $\Q$ is the uniform distribution.  Now, let us focus on $G'$ for the planted $\P_{H, L}$. We can describe the distribution $\P'_{H, L}$ of $G'$ as follows:
\begin{enumerate}
\item Choose a random subset $S'$ of $k - \ell$ vertices in $[n] \setminus L$.
\item Choose a random permutation $\pi'\colon S' \to [k] \setminus L$. Extend $\pi'$ to a permutation from $S' \cup L$ to $[k]$ by setting $\pi'(u) = u$ for all $u \in L$.
\item Set
\[ G'(u_1, \dots, u_r) = 
\begin{cases}
H(\pi'(u_1), \dots, \pi'(u_r)), &\text{if $u_1, \dots, u_r \in S' \cup L$} \\
\text{a random bit}, &\text{otherwise}.
\end{cases}
\]
\end{enumerate}

Using the above observations we have,
\begin{align*}
    \mathcal{LR}_D(H, L) &= \displaystyle\sup\limits_{\substack{p \in \R[G(\vec{u}) : \vec{u} \in \binom{[n]}{r}]\\ \deg p \le D}} \frac{\E_{\P_{H, L}}[p(G)] - \E_{\Q_{H, L}}[p(G)]}{\sqrt{\Var_{\Q_{H, L}}[p(G)]}}\\
    &= \displaystyle\sup\limits_{\substack{p \in \R[G'(\vec{u}) : \vec{u} \in \binom{[n]}{r} \setminus \binom{[\ell]}{r}]\\ \deg p \le D}} \frac{\E_{\P'_{H, L}}[p(G')] - \E_{\Q}[p(G')]}{\sqrt{\Var_{\Q}[p(G')]}}
\end{align*}


Since the null distribution $\Q$ is a product measure, by a standard linear algebraic argument in the literature of the low-degree method (see \cite{kunisky2019notes} or \cite[Lemma 7.2]{cojaoghlan2022statistical}), the optimal degree-$D$ polynomial takes an explicit form.
Using the expansion with respect to the Fourier-Walsh basis $\{\chi_\alpha(G')=\prod_{ e \in \alpha} G'_e, \alpha \subseteq \binom{[n]}{r} \setminus \binom{[\ell]}{r}\}$, the explicit formula for the squared advantage is
\begin{align}\label{eq:lr-fml}
\mathcal{LR}_D(H, L)^2 &= \sum_{\substack{\alpha \subseteq \binom{[n]}{r} \setminus \binom{L}{r} \\ 1 \leq \abs{\alpha} \leq D}}
\widehat{\mathcal{LR}}(\alpha | H, L)^2
\end{align}
where 
\[\widehat{\mathcal{LR}}(\alpha | H, L) = \E_{\Q}\frac{\P'_{H, L}(G')}{\Q(G')}\chi_\alpha(G')=\E_{\P_{H, L}'} \chi_\alpha(G'). \]
Now we expand the square on the right-hand side of \eqref{eq:lr-fml} and take the expectation over $H$.
\begin{equation}
\label{eq:replica}
\E_H \mathcal{LR}_D(H, L)^2 =   \sum_{\substack{\alpha \subseteq \binom{[n]}{r} \setminus \binom{L}{r} \\ 1 \leq \abs{\alpha} \leq D}} \E \chi_\alpha(G')\chi_\alpha(G''), 
\end{equation}
where the right-hand expectation is now taken over both the choice of $H$ and the choice of two independent ``replicas'' $G', G''$ sampled from $\P'_H$.  The joint distribution of $G'$ and $G''$ is determined by the independent choices of $H$, the subsets $S'$, $S''$, and the permutations $\pi'$, $\pi''$.
Equation \eqref{eq:replica} gives a formula for the second moment of the likelihood ratio with respect to the random variable $H$, which we spend the rest of this section evaluating; higher moments will be computed later.

We fix $\alpha \subseteq \binom{[n]}{r} \setminus \binom{L}{r}$ and upper bound the expectation. Since we are considering the expectation of a Fourier character, it will often be zero. Let $V(\alpha)$ be the set of vertices in $[n]$ spanned by $\alpha$.  If $S' \cup L$ or $S'' \cup L$ does not entirely contain $V(\alpha)$ then the expectation is zero: if, say, $e \in \alpha'$ is not contained in $S' \cup L$, then $G'(e)$ is independent of all other bits appearing in $\E \chi_\alpha(G')\chi_\alpha(G'') = \prod_{e \in \alpha} G'(e)G''(e)$ resulting in a value of zero.  Therefore
\begin{multline}
\label{eq:eprod}
\E[\chi_\alpha(G') \chi_\alpha(G'')] \\
\begin{aligned}
&= \E[\chi_\alpha(G') \chi_\alpha(G'')\ |\ S' \cap S'' \supseteq V(\alpha) \setminus L] \cdot \Pr[S' \cap S'' \supseteq V(\alpha) \setminus L] \\
&= \E[\chi_\alpha(G') \chi_\alpha(G'')\ |\ S' \cap S'' \supseteq V(\alpha) \setminus L] \cdot \Pr[S' \supseteq V(\alpha) \setminus L]^2
\end{aligned}
\end{multline}
by independence of $S'$ and $S''$.  As $S'$ is a random $k$-subset of $[n] \setminus L$, 
\begin{align}
\Pr[S' \supseteq V(\alpha) \setminus L] 
&= \frac{(k-\ell)(k-\ell-1)\cdots (k - \ell - |V(\alpha) \setminus L|+1)}{(n-\ell)(n-\ell - 1)\cdots(n-\ell - |V(\alpha) \setminus L|+1)} \notag \\
&\leq \left(\frac{k - \ell}{n - \ell}\right)^{\abs{V(\alpha) \setminus L}}\,. \label{eq:p1}
\end{align}

Conditioned on both $S'$ and $S''$ containing $V(\alpha) \setminus L$, 
\begin{align}
\chi_\alpha(G') \chi_\alpha(G'')
&= \prod_{(u_1, \dots, u_r) \in \alpha} G'(u_1, \dots, u_r) G''(u_1, \dots, u_r) \nonumber\\
&= \prod_{(u_1, \dots, u_r) \in \alpha} H(\pi'(u_1), \dots, \pi'(u_r)) H(\pi''(u_1), \dots, \pi''(u_r)). \label{eq:prod_H}
\end{align}
As $H$ consists of i.i.d. zero mean $\pm 1$ entries, this expression vanishes in expectation unless every hyperedge in the collection 
\[ (\psi(u_1), \dots, \psi(u_r))\colon (u_1, \dots, u_r) \in \alpha, \psi \in \{\pi', \pi''\} \]
appears exactly twice, in which case the product equals to one. This is only possible if $\pi\colon S' \to S''$ given by $\pi = (\pi'')^{-1} \circ \pi'$ restricts to an automorphism of $\alpha$.  In particular, $\pi$ must fix the set $V(\alpha)$.  As $\pi$ outside $L$ is a permutation which is chosen uniformly at random, we conclude that \eqref{eq:prod_H} is upper bounded by,
\begin{align} 
\Pr[\text{$\pi$ fixes $V(\alpha)$}] 
&= \frac{\abs{V(\alpha) \setminus L}!}{(k - \ell)(k - \ell - 1) \cdots (k - \ell - \abs{V(\alpha) \setminus L} + 1)} \notag \\
&\leq \biggl(\frac{\abs{V(\alpha) \setminus L}}{k - \ell}\biggr)^{\abs{V(\alpha) \setminus L}}. \label{eq:p2}
\end{align}
Plugging \eqref{eq:p1} and \eqref{eq:p2} into \eqref{eq:eprod} and then into \eqref{eq:replica} yields
\begin{equation}\label{eq:combinatorial-bound}
\E \mathcal{LR}_D(H, L)^2 \leq 
\sum_{\substack{\alpha \subseteq \binom{[n]}{r} \setminus \binom{L}{r} \\ 1 \leq \abs{\alpha} \leq D}} \biggl(\frac{\abs{V(\alpha) \setminus L}(k-\ell)}{(n - \ell)^2} \biggr)^{\abs{V(\alpha) \setminus L}}.
\end{equation}

This bound only depends on the hypergraph $\alpha$ through $\abs{V(\alpha) \setminus L}$. For $v=1,\ldots,rD$ let 
\begin{equation}\label{eq:nvd}
N(v,D) = \bigl|\bigl\{\alpha \subseteq \tbinom{[n]}{r} \setminus \tbinom{L}{r} : |V(\alpha) \setminus L| = v, \; 1 \le |\alpha| \le D\bigr\}\bigr|\,.
\end{equation}
Grouping the terms on the right-hand side by the value of $v = \abs{V(\alpha) \setminus L}$ gives
\begin{equation}
\label{eq:boundwithN}
\E \mathcal{LR}_D(H)^2 \leq \sum_{v=1}^{rD} N(v, D) \cdot 
\biggl(\frac{v(k-\ell)}{(n - \ell)^2} \biggr)^v.
\end{equation}  
To finish the proof we will demonstrate that this sum is dominated by the leading term $v = 1$.  We split this proof using the following two propositions.

In the first proposition, we bound the ``low'' vertex size part.
\begin{proposition}
\label{prop:low}
Assume that $e(k-\ell)/(n-\ell) \leq n^{-\epsilon}$. Then for every $0<\delta < \epsilon$ it holds for sufficiently large $n,$
\[ \sum_{v=1}^{\lfloor t \rfloor } N(v, D) \biggl(\frac{v(k-\ell)}{(n - \ell)^2} \biggr)^v \leq 2^{\binom{\ell}{r - 1}} \cdot 
\frac{n^{-\epsilon}}{1 - n^{-\epsilon + \delta}},
\]
where 
\begin{align}\label{eq:t}
    t := e^{-1}(r-1) (\delta \log n)^{1/(r-1)} - \ell
\end{align}
\end{proposition}

In the second proposition, we bound the ``high'' vertex size part.
\begin{proposition}
\label{prop:high}
Assume that $e(k-\ell)/(n-\ell) \leq n^{-\epsilon}$. Assume also that for some $\delta>0$ for which $0<\delta < \epsilon$, it holds
\begin{enumerate}
    
    \item  $\ell \leq (r/9) (\delta \log n)^{1/(r-1)}$
    
    and,

    \item $D \leq \epsilon \delta^2 (\log n)^{r/(r-1)}/\left(\frac{r}{r-1} \log \log n \right) $.
\end{enumerate}
Then for $t$ given in \eqref{eq:t} if also $\delta <1/4$ it holds,
\[ \sum_{v=\lfloor t \rfloor +1}^{rD} N(v, D) \biggl(\frac{v(k-\ell)}{(n - \ell)^2} \biggr)^v \leq \exp\left(-\Omega(\delta^{1/(r-1)}\epsilon r (\log n)^{r/(r-1)})\right).
\] 
\end{proposition}

Notice now that directly combining both the Propositions for $\delta=\epsilon/4$ directly implies Theorem \ref{thm:main}.

\subsection{Proof of Proposition \ref{prop:low}}
\begin{proof}
For fixed $v$, the set $V(\alpha) \setminus L$ can be chosen in $\binom{n - \ell}{v}$ ways.  The subset $\alpha$ can then include any
of the hyperedges in $V(\alpha)$ of which there are at most $\binom{v + \ell}{r}$, except those that at completely contained in $L$ of which there are $\binom{\ell}{r}$, leading to the bound:
\begin{equation}
\label{eq:Nsmall}
N(v, D) \leq \binom{n-\ell}{v} \cdot 2^{\binom{v + \ell}{r} - \binom{\ell}{r}}.
\end{equation}

Bounding $N(v, D)$ by~\eqref{eq:Nsmall} and using the standard binomial coefficient bound $\binom{a}{b} \leq (ea/b)^b$, the left hand side is at most
\[ \sum_{v=1}^t \biggl(\frac{e(k-\ell)}{n-\ell} \biggr)^v 2^{\binom{v + \ell}{r} - \binom{\ell}{r}} 
\]
As $e(k-\ell)/(n-\ell) \leq n^{-\epsilon} = n^{-\epsilon + \delta} \cdot n^{\delta}$, this is bounded by
\begin{equation}
\label{eq:bound1}
\sum_{v=1}^t n^{-(\epsilon - \delta)v} \cdot 2^{-\delta v \log_2 n + \binom{v + \ell}{r} - \binom{\ell}{r}}. 
\end{equation}
Let $f(v) = -\delta v \log_2 n + \binom{v + \ell}{r} - \binom{\ell}{r}, v \geq 1$.  For all integer $v \geq 1$,
\[ f(v+1)-f(v)=-\delta \log_2 n+\binom{v + \ell}{r-1} \leq -\delta \log n + \biggl(\frac{e(v + \ell)}{r - 1}\biggr)^{r-1}.
\]
By the definition of $t$, this is negative when $1 \leq v \leq t$, so $f(v)$ is maximized at $v = 1$. Therefore~\eqref{eq:bound1} is at most
\[ \sum_{v=1}^{\lfloor t \rfloor} n^{-(\epsilon - \delta)v} \cdot 2^{-\delta \log n + \binom{\ell + 1}{r} - \binom{\ell}{r}} \leq 2^{\binom{\ell}{r - 1}} \cdot 
\frac{n^{-\epsilon}}{1 - n^{-\epsilon + \delta}}
\]
using the identity $\binom{\ell + 1}{r} - \binom{\ell}{r} = \binom{\ell}{r-1}$ and the geometric sum formula. \hfill\qed
\end{proof}

\subsection{Proof of Proposition \ref{prop:high}}
\begin{proof}

When $v$ is large, the bound~\eqref{eq:Nsmall} can be improved by taking into account that at most $D$ of the hyperedges can be chosen:
\begin{align*}
    N(v, D) &\leq \binom{n-\ell}{v} \cdot D \binom{\binom{v + \ell}{r} - \binom{\ell}{r}}{D}\\
    &\leq D\biggl(\frac{e(n - \ell)}{v}\biggr)^v \cdot \biggl(\frac{e\binom{v + \ell}{r}}{D}\biggr)^D\\
 & \leq \biggl(\frac{e(n - \ell)}{v}\biggr)^v \cdot \biggl(\frac{e(v + \ell)}{r}\biggr)^{rD}
 \cdot D\biggl(\frac{e}{D}\biggr)^{D}.
\end{align*}
Under the assumption $e(k - \ell)/(n - \ell) \leq n^{-\epsilon}$ the summation of interest is at most
\begin{multline*} 
\sum_{v=t+1}^{rD} \biggl(\frac{e(k - \ell)}{n-\ell}\biggr)^v \cdot \biggl(\frac{e(v + \ell)}{r}\biggr)^{rD}
 \cdot D\biggl(\frac{e}{D}\biggr)^{D} \\
\begin{aligned}
& \leq rD^2\biggl(\frac{e}{D}\biggr)^{D} \cdot n^{-\epsilon t} \biggl(\frac{e(rD + \ell)}{r}\biggr)^{rD}\\
& \leq rD^2\biggl(\frac{e}{D}\biggr)^{D} \cdot n^{-\epsilon t} \bigl(e(D + \ell)\bigr)^{rD}\,. 
\end{aligned}
\end{multline*}
As $D \leq \epsilon \delta^2 (\log n)^{r/(r-1)}/(\frac{r}{r-1} \log \log n)$ and $\ell \leq (r/9) (\delta \log n)^{1/(r-1)}$, for sufficiently large $n$,
\begin{align*}
D \log \left((D+\ell)/(\epsilon \delta^2) \right) &\leq \epsilon \delta^2 (\log n)^{r/(r-1)},
\end{align*}Hence, for sufficiently small constant $0<\delta<1,$ for sufficiently large $n$ it holds
\begin{align*}
D \log \left(e(D+\ell)) \right) &\leq \epsilon \delta^2 (\log n)^{r/(r-1)},
\end{align*}
Using also the elementary inequality $D^2 (e/D)^D \leq 8$ we conclude that the summation of interest is at most
\begin{align*}
8rn^{-\epsilon t} \exp \left(\epsilon \delta^2 (\log n)^{r/(r-1)}\right).
\end{align*}
Plugging in the direct bound from the definition of $t$ and the upper bound on the leaked vertices,
\begin{align*}
t &\geq \frac{r}{7} (\delta \log n)^{1/(r-1)}
\end{align*}
we conclude that the summation of interest is at most
\begin{align*}
8r\exp \left(-\epsilon \frac{r-1}{e} \delta^{1/(r-1)} (\log n)^{r/(r-1)}+\epsilon \delta^2 (\log n)^{r/(r-1)}\right).
\end{align*}Choosing now $\delta<1/4$  concludes the result. \hfill\qed
\end{proof}


\subsection{Extension to higher moments}


Now we extend the calculation in Theorem \ref{thm:main} from $p =1$ to higher $p$.
The $2p$-th moment of $\mathcal{LR}_{D}(H, L)$ is
\begin{align*}
    \E_H \mathcal{LR}_{D}(H, L)^{2p} &= \E_H \Biggl(\sum_{\substack{\alpha \subseteq \binom{[n]}{r} \setminus \binom{L}{r} \\ 1 \leq \abs{\alpha} \leq D}} \E_{\substack{G' \sim \P'_H\\G'' \sim \P'_H}} \chi_\alpha(G')\chi_\alpha(G'')\Biggr)^p\\
    &= \sum_{\substack{\alpha_1, \dots, \alpha_p \subseteq \binom{[n]}{r} \setminus \binom{L}{r} \\ 1 \leq \abs{\alpha_i} \leq D}} \E \prod_{i=1}^p\chi_{\alpha_i}(G'_i)\chi_{\alpha_i}(G''_i)
\end{align*}
where the expectation is over $H$ and also over the replicas $G'_i, G''_i$ sampled independently from $\P'_H$.
Each $G'_i$ is equivalently sampled as $S'_i$ and $\pi_i'$ (and likewise $G''_i$ as $S''_i$ and $\pi''_i$).

Fix the Fourier characters $\alpha_1, \dots, \alpha_p$ and let $V(\alpha_i)$ be the set of vertices in $[n]$ spanned by $\alpha_i$.
First, the expectation is only nonzero if all of the sets $S'_i$ and $S''_i$ contain $V(\alpha_i) \setminus L$. By \eqref{eq:p1} this occurs with probability at most
\begin{equation}\label{eq:s-moment}
    \Pr\left[\forall i \in [p].\; S'_i \cap S''_i \supseteq V(\alpha_i) \setminus L\right] \leq \left(\frac{k-\ell}{n-\ell}\right)^{2\sum_{i = 1}^p |V(\alpha_i) \setminus L|}\,.
\end{equation}
Conditioned on this event,
\begin{align*}
    \prod_{i=1}^p\chi_{\alpha_i}(G'_i)\chi_{\alpha_i}(G''_i) &= 
    \prod_{i=1}^p \chi_{\pi'_i(\alpha_i)}(H) \chi_{\pi''_i(\alpha_i)}(H)\,.
\end{align*}
When the expectation is taken over $H$, this is only nonzero if every hyperedge appears an even number of times among the collection of edges
\[C := (\psi_i(u_1), \dots, \psi_i(u_r)) : i \in [p], \; (u_1, \dots, u_r) \in \alpha_i, \;\psi_i \in \{\pi'_i, \pi''_i\}\,.\]
In order for this to occur, every vertex in the image of the $\psi_i$ must be in the image of at least two $\psi_i$.
Let us say that the collection of embeddings is a \emph{double cover} if this occurs. Then
\begin{align}\label{eq:dc}
    &\quad \E_{H, \pi'_i, \pi''_i} \prod_{i=1}^p \chi_{\pi'_i(\alpha_i)}(H) \chi_{\pi''_i(\alpha_i)}(H) \nonumber \\
    &= \Pr_{\pi'_i, \pi''_i}[C \text{ is an even collection}] \nonumber\\
    &\leq \Pr_{\pi'_i, \pi''_i}[(\pi'_i, \pi''_i)_{i \in [p]} \text{ is a double cover}]\,.
\end{align}
Let $V = \sum_{i=1}^p |V(\alpha_i) \setminus L|$.
We claim
\begin{equation}\label{eq:dc-bound}
    \Pr_{\pi'_i, \pi''_i}[(\pi'_i, \pi''_i)_{i \in [p]} \text{ is a double cover}]
    \leq \frac{(2V)^{2V}}{(k-\ell)(k-\ell-1)\cdots (k-\ell-V+1)}\,. 
\end{equation}
This is based on the following surjection a.k.a union bound.
The total number of vertices mapped by all the permutations is $2V$.
We take any partition of the $2V$ vertices such that every block of the partition has size at least two. There are at most $(2V)^{2V}$ such partitions.
We go through the vertices in some fixed order, and for each vertex which
is not the first member of its block of the partition, we obtain a factor of ${\scriptstyle\approx}\frac{1}{k-\ell}$ for the probability that the vertex is mapped to the same element as the other members of its block of the partition.
Since the blocks have size at least two (in order to be a double cover),
we obtain at least $V$ factors of ${\scriptstyle\approx}\frac{1}{k-\ell}$ in this way.
We upper bound ${\scriptstyle \approx}\frac 1{k-\ell}$ by a rising factorial to obtain the bound in \eqref{eq:dc-bound}.

If $V \leq \frac{k-\ell}{2}$, then
\eqref{eq:dc-bound} can simplified to
\begin{align}
    \frac{(2V)^{2V}}{(k-\ell)(k-\ell-1)\cdots (k-\ell-V+1)} \leq \left(\frac{8V^2}{k - \ell}\right)^V\,.
\end{align}
On the other hand, if $V \geq \frac{k - \ell}{2}$, then the right-hand side is at least 1.
Combining these two possible cases, we conclude,
\begin{equation}\label{eq:dc-final}
        \Pr_{\pi'_i, \pi''_i}[(\pi'_i, \pi''_i)_{i \in [p]} \text{ is a double cover}] \leq \left(\frac{8V^2}{k-\ell}\right)^V\,.
\end{equation}

Now we return to the main calculation of $\E_H\mathcal{LR}_D(H, L)^{2p}$.
Combining \eqref{eq:s-moment}, \eqref{eq:dc-final},
\begingroup
\allowdisplaybreaks
\begin{align*}
    \E_H \mathcal{LR}_{D}(H, L)^{2p} &= \sum_{\substack{\alpha_1, \dots, \alpha_p \subseteq \binom{[n]}{r} \setminus \binom{L}{r} \\ 1 \leq \abs{\alpha_i} \leq D}} \E \prod_{i=1}^p\chi_{\alpha_i}(G'_i)\chi_{\alpha_i}(G''_i)\\
    &\leq \sum_{\substack{\alpha_1, \dots, \alpha_p \subseteq \binom{[n]}{r} \setminus \binom{L}{r} \\ 1 \leq \abs{\alpha_i} \leq D}} \left(\frac{8V^2(k-\ell)}{(n-\ell)^2}\right)^V \\
    &\leq \sum_{\substack{\alpha_1, \dots, \alpha_p \subseteq \binom{[n]}{r} \setminus \binom{L}{r} \\ 1 \leq \abs{\alpha_i} \leq D}} \left(\frac{8p^2D^2(k-\ell)}{(n-\ell)^2}\right)^{\sum_{i=1}^p |V(\alpha_i)\setminus L|} & (V \le pD)\\
    &= \Biggl(\sum_{\substack{\alpha\subseteq \binom{[n]}{r} \setminus \binom{L}{r} \\ 1 \leq \abs{\alpha} \leq D}} \left(\frac{8p^2D^2(k-\ell)}{(n-\ell)^2}\right)^{|V(\alpha)\setminus L|}\Biggr)^p
\end{align*}
\endgroup
The inner summation is nearly the combinatorial quantity we bounded in Equation \eqref{eq:combinatorial-bound} when computing $\E_H\mathcal{LR}_D(H, L)^2$.
The only difference is the factor $8p^2D^2$ which may be larger than what we had before.
This factor can be negated by scaling down $\frac{k-\ell}{n-\ell}$. Using the same counting arguments as before with the slightly stronger assumption on $k$, we conclude the desired moment bound.

\section{Cryptographic applications}

\subsection{Hypergraph secret sharing}
\label{sec:secretsharing}

The secret sharing scheme of Abram et al. was stated for forbidden graph access structures.  The construction extends to partial access structures $(R, S)$ where $R$ is a collection of $r$-subsets and $S$ consists of all independent sets of $R$ of size at most $\ell$.  

\begin{construction} Forbidden hypergraph secret sharing:
\label{cons:hypergraph} Syntactically replace ``graph'' by ``$r$-uniform hypergraph'' and $(u, v)$ by $(u_1, \dots, u_r)$ in Construction~\ref{cons:abrametal}.
\end{construction}

This scheme reconstructs all $\{u_1, \dots, u_r\} \in R$ by~\eqref{eq:reconstruct}.  

\begin{proposition}
Assume $(H, \P_{H,L})$ and $(H, \Q_{H,L})$ are $(s, \epsilon)$-indistinguishable for all $L \subseteq V(H)$ with $|L| = \ell$. Then for every independent set $I \subseteq R$ of size at most $\ell$, shares of $0$ and $1$ are $(s, 2\epsilon)$-indistinguishable by parties in $I$.
\end{proposition}
\begin{proof}
Assume parties in $I$ can $2\epsilon$-distinguish shares of 0 and 1 using distinguisher $D$.  By the triangle inequality, $D$ $\epsilon$-distinguishes $(H_s, G, \phi(i): i \in I)$ from $(H, G,\phi(i): i \in I)$ where 
\[ G(u_1, \dots, u_r) = \begin{cases}
H(u_1, \dots, u_r), &\text{if $u_1, \dots, u_r \in I$} \\
\text{a random bit}, &\text{otherwise}.
\end{cases}
\]
for at least one value of $s$.  Let $D'$ be the circuit that, on input $(H', G, u_i \colon i \in I)$, outputs $D(H' \oplus sR,  G, u_i \colon i \in I)$. As $R$ does not contain any hyperedges within $I$, by~\eqref{eq:defhs},
$D'(\mathcal{P}_{H, I})$ is identically distributed to $D(H_s, G, \phi(i)\colon i \in I)$.  As $H$ is random, $D'(\mathcal{Q}_{H, I})$ is identically distributed to $D(H, G, \phi(i)\colon i \in I)$.  Therefore $D'$ and $D$ have the same advantage.   \hfill\qed
\end{proof}

The class of access structures can be expanded to allow the reconstruction set $R$ to consist of arbitrary sets, as long as the size of all minimal sets is at most $r$.  This is accomplished by a reduction to size exactly $r$. Let $R' \subseteq [n + r - 1]$ be the $r$-uniform hypergraph
\[ R' = \bigl\{A \cup \{n + 1, \dots, n + r - \abs{A}\}\colon A \in R\bigr\}. \]

\begin{construction}
\label{cons:atmostr}
Apply Construction~\ref{cons:hypergraph} to $R'$ with the shares of parties $n + 1, \dots, n + r - 1$ made public.
\end{construction}

If all sets in $R'$ can resconstruct in Construction~\ref{cons:hypergraph} then all sets in $R$ can reconstruct in Construction~\ref{cons:atmostr}.  As for secrecy, if Construction~\ref{cons:hypergraph} is secure against all independent sets in $R$ of size at most $\ell$, then Construction~\ref{cons:atmostr} is secure against such sets of size at most $\ell - r + 1$.

\medskip\noindent
Could Construction~\ref{cons:hypergraph} give a \emph{provable} separation between the minimum share size of information-theoretic and computational secret sharing? We argue that this is unlikely barring progress in information-theoretic secret sharing lower bounds.  The share size in Construction~\ref{cons:hypergraph} is $(1 + \Omega(1))(\log n)$.  However, the share size lower bounds of~\cite{kilian-nisan, bogdanov-guo-komargodski} do not exceed $\log n$ for any known $n$-party access structure.

In contrast, Csirmaz~\cite{csirmaz} proved that there exists an $n$-party access structure with share size $\Omega(n/\log n)$.  Using Csirmaz's method,  Beimel~\cite{beimel} constructed \emph{total} $r$-hypergraph access structures that require share size $\Omega(n^{2 - 1/(r-1)}/r)$ for every $r \geq 3$.  

We argue that Csirmaz's method cannot prove a lower bound exceeding $\ell$ for any (partial) access structures in which secrecy is required to hold only for sets of size up to $\ell$.  Csirmaz showed that a scheme with share size $s$ implies the existence of a monotone submodular function $f$ (the joint entropy of the shares in $A$) from subsets of $\{1, \dots, n\}$ to real numbers that satisfies the additional constraints
\begin{align}
f(A) + f(B) &\geq f(A \cup B) + f(A \cap B) + 1 &&\text{if $A, B \in S$ and $A \cup B \in R$} \label{eq:extramodular} \\
f(A) &\leq s &&\text{for all $A$ of size $1$}. \label{eq:minimal}
\end{align}

\begin{proposition}
\label{prop:infupper}
Assuming all sets in $S$ have size at most $\ell$, there exists a monotone submodular function satisfying~\eqref{eq:extramodular} and~\eqref{eq:minimal} with $s = \ell$.
\end{proposition}

As our scheme does not tolerate $\Omega(\log n)$ bits of leakage, the best share size lower bound that can be proved using Csirmaz's relaxation of secret sharing is $\ell = o(\log n)$. The proof of Proposition \ref{prop:infupper} is a natural generalization of~\cite[Theorem~3.5]{csirmaz} to partial access structures.

\begin{proof}[Proposition \ref{prop:infupper}]
The function $f(A) = \sum_{t=1}^{\abs{A}} \max\{\ell - t + 1, 0\}$ is monotone, submodular, satisfies~\eqref{eq:extramodular} for every $R \subseteq \overline{S}$, and~\eqref{eq:minimal} with $s = \ell$. \hfill\qed
\end{proof}

\subsection{Multiparty PSM for random functions}
\label{sec:psm}

Given a function $F\colon [k]^r \to \{\pm 1\}$, the random hypergraph embedding of $F$ is the $r$-hypergraph $\overline{F}$ on $rk$ vertices $(x, i) \colon x \in [k], i \in [r]$ such that \[ \overline{F}((x_1, 1), \dots, (x_r, r)) = F(x_1, \dots, x_r). \]  All other potential hyperedges of $\overline{F}$ are sampled uniformly and independently at random.

We describe the $r$-partite generalization of Abram et al.'s PSM protocol.  Let $\phi\colon [k] \times [r] \to [n]$ be a random injection and let $G$ be the $r$-hypergraph on $n$ vertices given by
\[ G(u_1, \dots, u_r) = \begin{cases}
\overline{F}(\phi^{-1}(u_1), \dots, \phi^{-1}(u_r)), &\text{if $\phi^{-1}(u_1), \dots, \phi^{-1}(u_r)$ exist} \\
\text{a random bit}, &\text{otherwise.}
\end{cases}
\]

\begin{construction}
\label{cons:psm}
$r$-party PSM protocol for $F$:
\begin{description}
\item In the setup phase, $G$ is published and $\phi$ is privately given to the parties.
\item In the evaluation phase, 
\begin{enumerate}
\item[1.] Party $i$ is given input $x_i$.
\item[2.] Party $i$ forwards $u_i = \phi(x_i, i)$ to the evaluator. 
\item[3.] The evaluator outputs $G(u_1, \dots, u_r)$.
\end{enumerate}
\end{description}
\end{construction}

The protocol is clearly functional.  A reasonable notion of security with respect to random functions $F$ should allow the parties' input choices to depend on $F$.  An \emph{input selector} is a randomized function $I$ that, on input $F$, produces inputs $I(F) = (x_1, \dots, x_r)$ for the $r$ parties.

We say a protocol is $(s', s, \epsilon)$ (simulation) secure against a random function if for every input selector $I$ there exists a size-$s'$ simulator $S$ for which the distributions 
\begin{equation}
\label{eq:psm}
(F, G, \phi(x_1, 1), \dots, \phi(x_r, r)) \quad\text{and}\quad (F, S(F, F(x_1, \dots, x_r)))
\end{equation}
are $(s, \epsilon)$-indistinguishable, where $(x_1, \dots, x_r)$ is the output of $I(F)$.

\begin{proposition}
Assume $(H, G, L(H))$ with $G \sim \P_{H,L(H)}$ versus $G \sim \Q_{H,L(H)}$ are $(s, \epsilon)$-indistinguishable with parameters $\abs{V(H)} = kr$, $\abs{V(G)} = n$, and $\ell = r$.  Then Construction~\ref{cons:psm} is $(O(\binom{n}{r}), s - O(\binom{n}{r}), \epsilon)$-secure.
\end{proposition}

We label the vertices of $H$ by pairs $(x, r) \in [k] \times [r]$.

\begin{proof}
On input $(F, y)$, the simulator $S$ 
\begin{enumerate}
\item chooses random $u_1, \dots, u_r \in [n]$
\item sets $G(u_1, \dots, u_r) = y$ 
\item samples all other possible  hyperedges of $G$ independently at random
\item outputs $(G, u_1, \dots, u_r)$.
\end{enumerate}

We describe a reduction $R$ that, given a distinguisher $D$ for \eqref{eq:psm}, tells apart $(H, G, L(H))$  with $G \sim \P_{H,L(H)}$ versus $G \sim \Q_{H,L(H)}$ for some leakage function $L$.  On input $(H, G, z_1, \dots, z_r)$, 
\begin{enumerate}
\item set $F$ to be the function $F(x_1, \dots, x_r) = H((x_1, 1), \dots, (x_r, r))$
\item output $(F, \pi(G), \pi(z_1), \dots, \pi(z_r))$ for a random permutation $\pi$ on $[n]$ (which acts on $G$ as a hypergraph isomorphism).
\end{enumerate}

Let $L$ be the leakage function that, on input $H$, runs $I(F)$ to obtain $(x_1, \dots, x_r)$, and outputs $((x_1, 1), \dots, (x_r, r))$. 

This reduction preserves distinguishing advantage as it maps the distributions~\eqref{eq:psm} into the distributions
$(H, \P_{H,L(H)}, L(H))$ and $(H, \Q_{H,L(H)}, L(H))$, respectively.   It can be implemented in size $O(\binom{n}{r})$, giving the desired parameters. \hfill\qed
\end{proof}


\bibliographystyle{alpha}
\bibliography{bibliography}

\end{document}